\documentclass[conference]{IEEEtran}
\IEEEoverridecommandlockouts
\usepackage{cite}
\usepackage{amsmath,amssymb,amsfonts}
\usepackage{graphicx}
\usepackage{textcomp}
\usepackage{xcolor}
\usepackage{pgfplots}
\usepackage{pgfplotstable}
\usetikzlibrary{patterns,patterns.meta}
\usepackage{tikz}
\usepgfplotslibrary{groupplots}  
\usepackage{pgfplots}
\usetikzlibrary{patterns, positioning}
\pgfplotsset{compat=1.18}
\usetikzlibrary{intersections, backgrounds, patterns}
\usepgfplotslibrary{fillbetween}
\usepackage{algorithm}
\usepackage{pgf-pie}  
\usepackage{algorithmic}
\usepackage{amsmath}
\usepackage{graphicx}
\usepackage{subcaption}
\usepackage{multirow}
\usepackage{xcolor}
\usepackage{booktabs}

\usetikzlibrary{shadows,patterns,decorations.pathreplacing}
\usepackage{amsthm}
\theoremstyle{plain}  
\newtheorem{theorem}{Theorem}
\newtheorem{definition}{Definition}

\definecolor{myblue}{RGB}{0,114,178}       
\definecolor{mygreen}{RGB}{0,158,115}      
\definecolor{myorange}{RGB}{230,159,0}     
\definecolor{myred}{RGB}{213,94,0}         
\definecolor{mypurple}{RGB}{86,180,233}    
\definecolor{myyellow}{RGB}{240,228,66}    
\definecolor{mycyan}{RGB}{0,158,115}       
\definecolor{mygray}{RGB}{114,114,114}     

\usepackage{newtxtext,newtxmath}
\usetikzlibrary{positioning,arrows.meta,backgrounds}

\def\BibTeX{{\rm B\kern-.05em{\sc i\kern-.025em b}\kern-.08em
    T\kern-.1667em\lower.7ex\hbox{E}\kern-.125emX}}
\begin{document}

\title{\emph{FedBit}: Accelerating Privacy-Preserving Federated Learning via Bit-Interleaved Packing and Cross-Layer Co-Design}

\author{\IEEEauthorblockN{Xiangchen Meng and Yangdi Lyu}
\IEEEauthorblockA{\textit{Microelectronics Thrust} \\
\textit{The Hong Kong University of Science and Technology (Guangzhou)}\\
Guangzhou, China \\
xmeng027@connect.hkust-gz.edu.cn, yangdilyu@hkust-gz.edu.cn}
}
\vspace{-27pt}

\maketitle

\begin{abstract}
Federated learning (FL) with fully homomorphic encryption (FHE) effectively safeguards data privacy during model aggregation by encrypting local model updates before transmission, mitigating threats from untrusted servers or eavesdroppers in transmission. However, the computational burden and ciphertext expansion associated with homomorphic encryption can significantly increase resource and communication overhead. To address these challenges, we propose \emph{FedBit}, a hardware/software co-designed framework optimized for the Brakerski-Fan-Vercauteren (BFV) scheme. \emph{FedBit} employs bit-interleaved data packing to embed multiple model parameters into a single ciphertext coefficient, thereby minimizing ciphertext expansion and maximizing computational parallelism. Additionally, we integrate a dedicated FPGA accelerator to handle cryptographic operations and an optimized dataflow to reduce the memory overhead. Experimental results demonstrate that FedBit achieves a speedup of two orders of magnitude in encryption and lowers average communication overhead by 60.7\%, while maintaining high accuracy.

\end{abstract}


\section{Introduction}

The remarkable successes of deep neural networks are heavily reliant on large amounts of high-quality data. The significant value and sensitivity of domain-specific data compel organizations to enforce strict protections to preserve confidentiality and maintain a competitive advantage~\cite{truong2021privacy}.
FL enables a collaborative model training across distributed data silos without exchanging raw data, thereby adhering to privacy mandates such as GDPR~\cite{gdpr2016} and HIPAA~\cite{hipaa1996}. However, despite its privacy-preserving design, FL remains susceptible to attacks such as membership inference, property inference, and gradient inversion~\cite{melis2019exploiting,nasr2018comprehensive}, which can leak sensitive client information through unprotected model updates. 
To mitigate these risks, cryptographic techniques, particularly FHE, have gained significant attention for their capability to perform computations directly on encrypted data. Prominent FHE schemes, such as BGV~\cite{zhang2023advancing}, CKKS~\cite{pan2024fedshe}, and BFV~\cite{cai2023secfed}, enable privacy-preserving aggregation in FL by encrypting local updates before transmission, and ensure confidentiality across the training process.

In popular FHE schemes used for FL, such as BFV and CKKS, data is encrypted into polynomials of degree greater than \(2^{12}\), with plaintext values packed into polynomial coefficients under a large modulus \(q\), typically ranging from 128 to 256 bits to meet modern security standards~\cite{gentry2009fully}. These characteristics introduce two primary challenges that limit the scalability of FHE-based FL:

\begin{enumerate}
    \item \textbf{Slow Encryption Speed}. FHE encryption and decryption involve expensive polynomial arithmetic, particularly multiplication, significantly slows encryption and decryption. Despite optimizations such as the number-theoretic transform (NTT), performance remains limited. For instance, CKKS encryption using Microsoft SEAL~\cite{sealcrypto} requires 482.68 seconds for ResNet-18 and 54.69 seconds for AlexNet~\cite{10.1145/3698038.3698557}.
    
    \item \textbf{High Communication Overhead.} FHE ciphertexts are significantly larger than plaintexts, leading to substantial bandwidth consumption during model aggregation. For instance, training on Fashion-MNIST and CIFAR-10 with only four clients required 58.7 GB and 227.8 GB of data transfer, respectively~\cite{zhang2020batchcrypt}. This communication burden severely limits the scalability of FHE-based FL, especially for large models and client populations.
\end{enumerate}

To address these challenges, the research community has actively pursued a variety of solutions for efficiency and scalability in privacy-preserving FL. 

On the software side,  Zhang et al.~\cite{zhang2020batchcrypt} tackled communication inefficiencies through batch encoding and gradient clipping, while the FedSHE framework~\cite{pan2024fedshe} introduces adaptive segmented CKKS encryption to balance security, efficiency, and numerical accuracy. Zuo et al.~\cite{10.1145/3698038.3698557} proposed error correction techniques based on fitting functions and deviation-weight filtering to mitigate encryption-induced errors during aggregation. In parallel, hardware accelerators have been explored to speed up arithmetic in homomorphic encryption.  HI-CKKS~\cite{cryptoeprint:2024/1976} achieves a $200\times$ speedup for NTT operation on GPUs. Similarly, Paul et al.~\cite{9857557} developed an FPGA-based accelerator for BGV encryption and decryption, achieving a $52.71\times$ speedup over conventional software implementations. Despite these advances, many existing system-level FL frameworks still rely on traditional partially homomorphic encryption schemes such as RSA~\cite{10374366} and Paillier~\cite{9893091}. While these schemes offer lightweight encryption suitable for secure aggregation, they lack support for SIMD-style data packing and require one ciphertext per model parameter, resulting in significant communication overhead when extended to large-scale model training

In this paper, we propose \emph{FedBit}, a novel approach to accelerate privacy-preserving FHE-FL from both communication and computation across the system stack. Our contributions include:

\begin{enumerate}
    \item We introduce a bit-interleaved packing scheme that shrinks ciphertexts and reduces computation times.
    \item We implement operator-level accelerators and a carefully scheduled dataflow to fully exploit on-chip resources for homomorphic primitives.

    \item The resulting end-to-end system delivers two orders of magnitude speedup on encryption/decryption and reduces per-client traffic by average \(\mathbf{60.7\%}\).
\end{enumerate}


\section{Preliminary}
\label{sec:preliminary}

\subsection{Brakerski-Fan-Vercautern Fully Homomorphic Encryption}


The BFV scheme~\cite{brakerski2014leveled} is a prominent FHE construction based on the Ring Learning With Errors (RLWE) problem. 
In the BFV framework, a plaintext message \(m\) is first encoded as a polynomial within the plaintext space \(R_t = \mathbb{Z}_t[X]/(X^N + 1)\), where \(t\) is the plaintext modulus and \(N\) is typically a power of two. This polynomial is represented as:
\begin{equation}
    m(X) = \sum_{i=0}^{N-1} m_i X^i \in R_t,
\end{equation}
with each coefficient \(m_i\) being an integer modulo \(t\).

Subsequently, for encryption, the plaintext polynomial \(m(X)\) is transformed into a ciphertext \(ct = (c_0, c_1)\), where each component lies
in the ciphertext space \(R_q = \mathbb{Z}_q[X]/(X^N + 1)\). The ciphertext is constructed as:
\begin{equation}
    \label{eq: enc and dec}
    ct = (c_0, c_1) = \left( a \cdot s + e + \Delta \cdot m(X), -a \right) \in R_q^2,
\end{equation}
where the scaling factor \(\Delta = \lfloor q/t \rfloor\), the secret key \(s\) is uniformly sampled from \(R_2\), \(a\) is uniformly sampled from \(R_q\), and \(e\) is a small error polynomial. For the decryption process, the plaintext \(m(X)\) is recovered from the ciphertext \(ct\) using the secret key \(s\):
\begin{equation}
    m(X) = \left\lfloor \frac{c_0 + c_1 \cdot s}{\Delta} \right\rceil \pmod{t}.
\end{equation}

Given two ciphertexts \( ct = (c_0, c_1) \) encrypting plaintext \( m \) and \( ct' = (c_0', c_1') \) encrypting plaintext \( m' \), their homomorphic sum \( ct_{\text{add}} \) is computed as:
\begin{equation}
    ct_{\text{add}} = (c_0 + c_0' \pmod{q}, c_1 + c_1' \pmod{q}).
\end{equation}
This ciphertext \( ct_{\text{add}} \) decrypts to \( m + m' \pmod{t} \).

\subsection{Federated Learning with Fully Homomorphic Encryption}

Federated learning involves a system that consists of a central server and multiple clients, each holding private datasets. To define the security framework, we introduce the following key entities:

\begin{definition}[Honest-but-Curious Server]
    The server adheres to the federated learning protocol but may attempt to infer additional information from the model updates received during aggregation, such as reconstructing individual clients' data or extracting sensitive patterns from the global model.
\end{definition}


\begin{definition}[
Adversarial Communication Channel]
    The communication channel between clients and the server is vulnerable to passive adversaries who can intercept all transmitted ciphertexts. The system must ensure confidentiality against such eavesdropping attacks.
        
\end{definition}

Under these security assumptions, the operational workflow of FL enhanced by FHE proceeds as illustrated in Figure~\ref{fig 2: Federated Learning with Fully Homomorphic Encryption}. The process ensures privacy preservation as follows:

\begin{figure}[t]
    \centering{\includegraphics[width=0.9\linewidth]{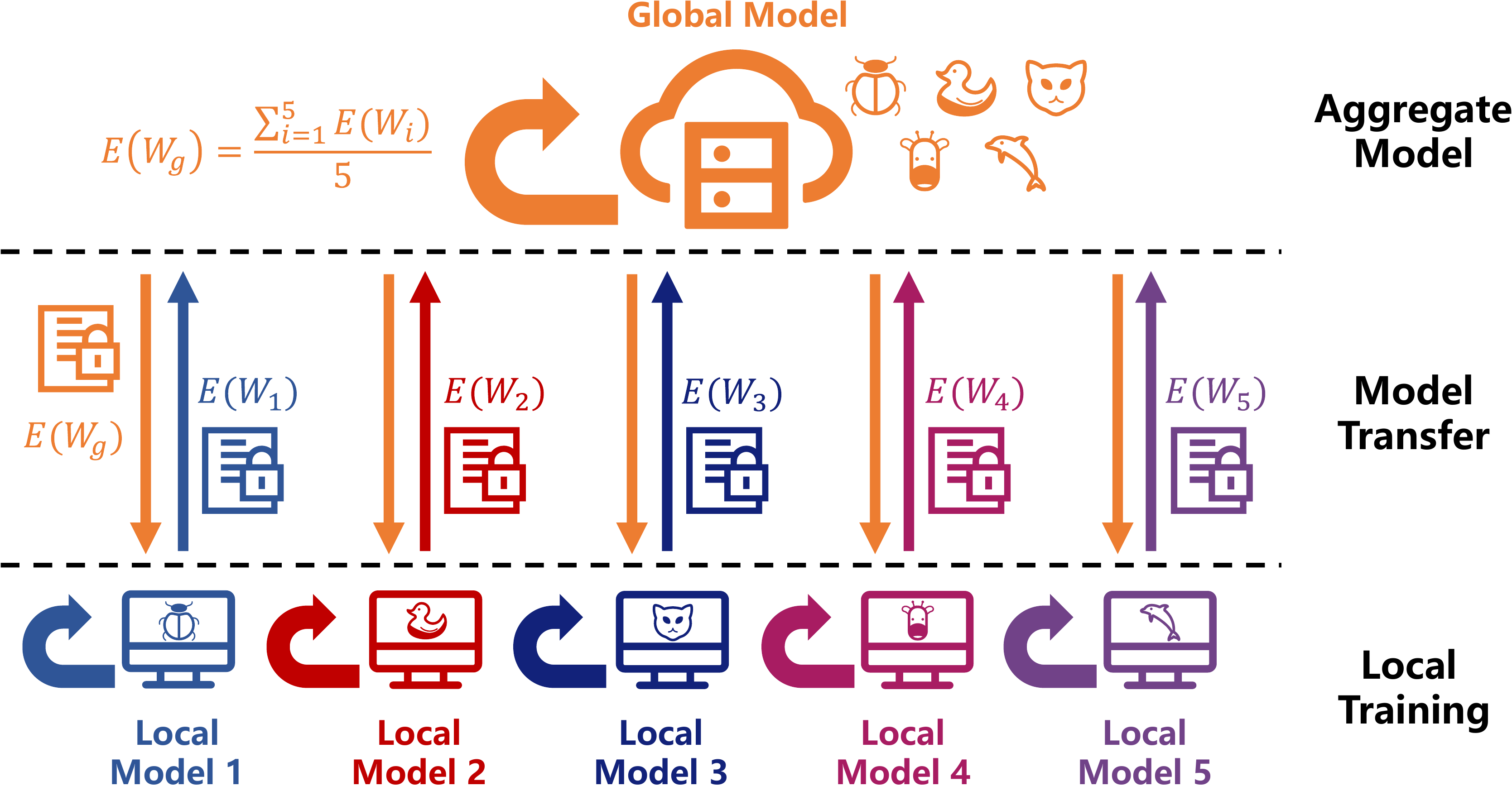}}
    \caption{FL with FHE Workflow}
    \label{fig 2: Federated Learning with Fully Homomorphic Encryption}
\end{figure}

\begin{enumerate}
    \item \textbf{Local Training}: Each client independently trains a local model using its private dataset, computing model updates (weights or gradients) without exposing raw data.
    
    \item \textbf{Encryption and Transmission}: Clients encrypt their model updates using a shared FHE scheme (with the same private key) before transmission, ensuring that the \emph{honest-but-curious server} cannot access plaintext updates.
    
    \item \textbf{Homomorphic Aggregation}: The server performs aggregation operations directly on the encrypted updates from multiple clients.
    
    \item \textbf{Distribution and Decryption}: The encrypted aggregated model is distributed back to clients through \emph{adversarial communication channels}. Each client decrypts the model using its private key to obtain the updated global model.
\end{enumerate}

This iterative cycle continues for multiple rounds until model convergence, with the defined security entities ensuring privacy throughout the entire process.
\section{Bit-Interleaved Packing Method}
\label{subsec:bit_interleaved_packing}

We propose a bit-interleaved packing strategy for BFV encryption that embeds multiple quantized neural network parameters into individual polynomial coefficients while preserving additive homomorphism.

\subsection{Layer-Specific Bit-Field Construction}

Our framework supports \emph{layer-specific quantization}, enabling varying precision levels for different layers of the neural network. For the $\ell^{\text{th}}$ layer, which consists of $r_\ell$ quantized weights, we pack $m_\ell$ quantized weights into each polynomial coefficient. The packing capacity, $m_\ell$, is determined by the constraints of bit-width. The bit-width allocated for each weight is defined by the quantized bits $\beta_\ell$ along with a carry-protection margin $\delta_\ell$. Each weight is assigned a slot of \((\beta_\ell + \delta_\ell)\) bits within a coefficient. Figure~\ref{fig:Bit-Interleaved Packing} demonstrates the bit-interleaved packing scheme with $m_\ell = 2$, utilizing 8-bit representations for quantized weights and a carry margin of $\delta_\ell = 2$ bits. To illustrate this method, consider two quantized weight values, $w_1 = 0$ and $w_2 = 9$. By applying bit-shifting operations, we compute $c = w_1 \cdot 2^{\beta_\ell + \delta_\ell} + w_2 = 0 \cdot 2^{10} + 9$, resulting in a packed polynomial coefficient of $c = 9216$.

\begin{figure}[t]

    \centering
    \includegraphics[width=1\columnwidth]{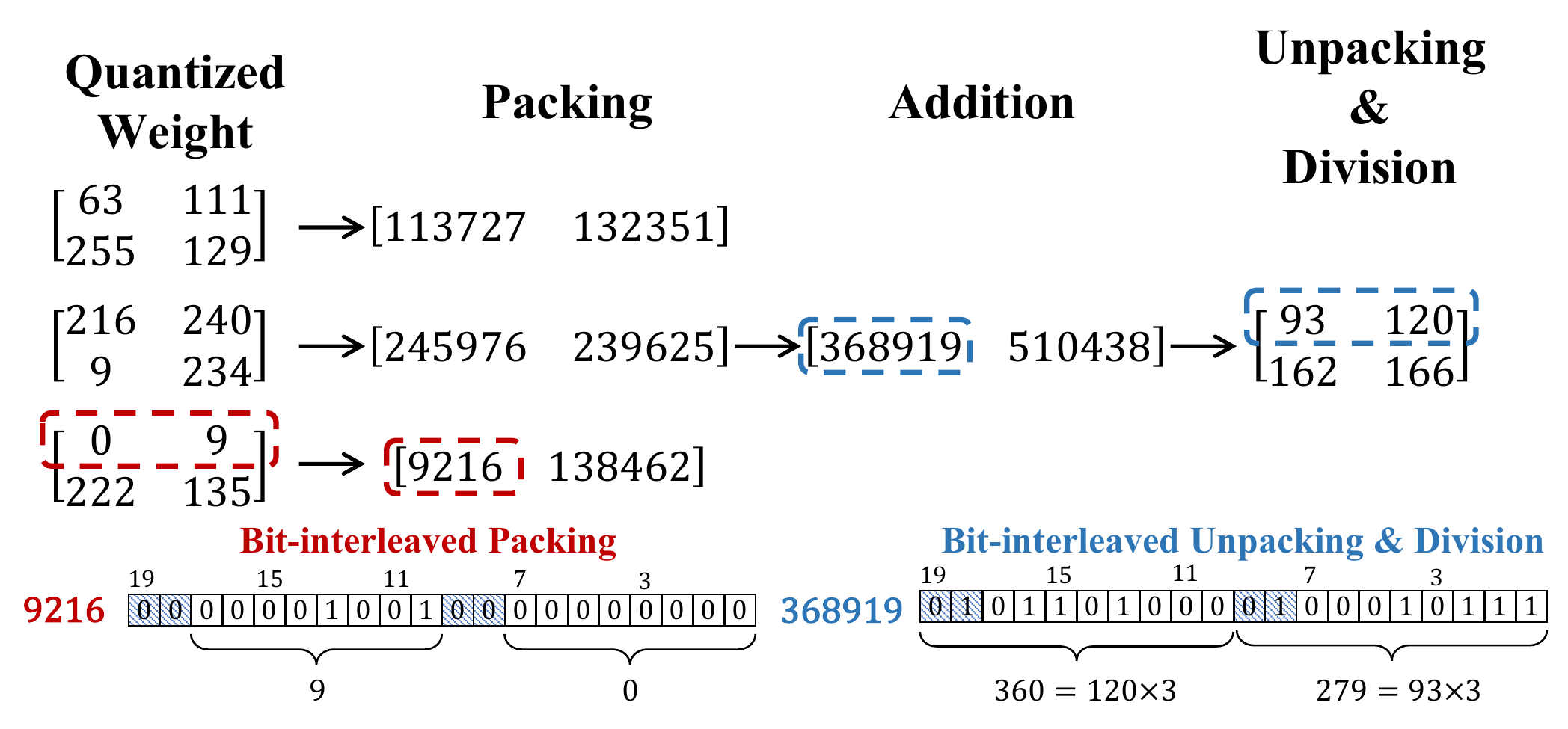}
    \caption{Bit-interleaved packing example with \(\beta_\ell = 8\), and \(\delta_\ell = 2\). In the Unpacking and Division stage, the data are divided by 3 to generate the average result.}
    \label{fig:Bit-Interleaved Packing}
    \vspace{-15pt}
\end{figure}


With $N$ coefficients, a single polynomial can encode $m_\ell \cdot N$ weights. Thus, the total number of plaintext polynomials needed for the $\ell^{\text{th}}$ layer is
\begin{equation}
    T_\ell = \left\lceil \frac{r_\ell}{m_\ell \cdot N} \right\rceil.
\end{equation}
The complete set of layer $\ell$ weights is encoded into $T_\ell$ plaintext polynomials:

\begin{equation}
    \label{eq:layer_polynomial}
    \mathbf{P}^{(\ell)} = \left\{ P_0^{(\ell)}(X), P_1^{(\ell)}(X), \ldots, P_{T_\ell - 1}^{(\ell)}(X) \right\}.
\end{equation}

Each polynomial $P_i^{(\ell)}(X)$ is constructed from $N$ packed coefficients, which represent a total of $m_\ell \cdot N$ quantized weights. To guarantee correct aggregation, we establish a fixed allocation policy that designates weights to coefficients.
\begin{equation}
    \label{eq:layer_polynomial_component}
    \begin{split}
    P_i^{(\ell)}(X) 
    &= \sum_{j=0}^{N-1}
        c_{i, j}^{(\ell)} X^j \\
    &= \sum_{j=0}^{N-1} \left( 
        \sum_{k=0}^{m_\ell - 1} 
        w_{(i N + j) m_\ell + k}^{(\ell)} \cdot 2^{k (\beta_\ell + \delta_\ell)} 
    \right) X^j
    \end{split}
\end{equation}

\subsection{Aggregation of Packed Polynomials}

During federated aggregation, the server computes the sum of $U$ encrypted model updates, each of which is encoded using the bit-interleaved packing method described in Equation~\eqref{eq:layer_polynomial_component}. Let $\mathbf{P}^{(\ell,u)} = \{P_0^{(\ell,u)}(X), \ldots, P_{T_\ell - 1}^{(\ell,u)}(X)\}$ denote the set of packed plaintext polynomials corresponding to the $\ell^{\text{th}}$ layer in the $u^{\text{th}}$ encrypted input. The homomorphic aggregation is performed at the polynomial level as
\begin{equation}
\label{eq:polynomial_aggregation}
\mathbf{P}^{(\ell)}_{\text{agg}} = \sum_{u=1}^{U} \mathbf{P}^{(\ell,u)} := \left\{
\sum_{u=1}^{U} P_0^{(\ell,u)}(X), \ldots,
\sum_{u=1}^{U} P_{T_\ell - 1}^{(\ell,u)}(X)
\right\}.
\end{equation}
After obtaining the aggregated polynomials \(\mathbf{P}^{(\ell)}_{\text{agg}}\), the next step is to compute the average model update by performing a scalar division of each packed coefficient by \(U\).

Figure~\ref{fig:Bit-Interleaved Packing} shows the aggregation phase with $U = 3$ clients. The accumulated coefficient value becomes $c_{\text{agg}} = 368919$. The unpacking procedure extracts the individual aggregated weights as $w_1^{\text{agg}} = 360$ and $w_2^{\text{agg}} = 279$. Subsequently, these values are normalized by the number of participating clients to obtain the averaged weights: $\bar{w}_1 = w_1^{\text{agg}}/U = 120$ and $\bar{w}_2 = w_2^{\text{agg}}/U = 93$.

\subsection{Aggregation Bounds for Correctness}
In order to ensure the correctness of the aggregation, we need to satisfy two key conditions: intra-slot carry prevention and inter-slot modulus constraint.

\begin{theorem}[Intra-Slot Carry Prevention]
\label{thm:intra-poly}
Bit-field isolation is preserved during aggregation of $U$ ciphertexts if:
\begin{equation}
\label{eq:poly_intra_bound}
U \cdot (2^{\beta_\ell} - 1) < 2^{\beta_\ell + \delta_\ell}.
\end{equation}
\end{theorem}

\begin{proof}
Each quantized weight $w_k^{(\ell)} \in [0, 2^{\beta_\ell}-1]$. Under worst-case conditions where all $U$ clients have maximum-valued weights, the aggregated value equals $U \cdot (2^{\beta_\ell} - 1)$. To prevent carry propagation into adjacent bit-fields, this sum must not exceed the allocated bit capacity $2^{\beta_\ell + \delta_\ell}$. The strict inequality ensures carry-free operation.
\end{proof}

\begin{theorem}[Inter-Slot Modulus Constraint]
\label{thm:inter-poly}
Define the maximum packed coefficient value as:
\begin{equation}
\label{eq:max_coefficient}
\mathcal{M}^{(\ell)} = \sum_{k=0}^{m_\ell-1} (2^{\beta_\ell} - 1) \cdot 2^{k \cdot (\beta_\ell + \delta_\ell)} = (2^{\beta_\ell} - 1) \cdot \frac{2^{m_\ell (\beta_\ell + \delta_\ell)} - 1}{2^{\beta_\ell + \delta_\ell} - 1}.
\end{equation}
Correct decryption after aggregation requires:
\begin{equation}
\label{eq:poly_inter_bound}
U \cdot \mathcal{M}^{(\ell)} < t.
\end{equation}
\end{theorem}

\begin{proof}
The maximum coefficient value $\mathcal{M}^{(\ell)}$ occurs when all $m_\ell$ packed weights equal $2^{\beta_\ell} - 1$. The closed-form expression follows from the geometric series summation. After aggregating $U$ such maximum-valued coefficients, the result must remain within the plaintext modulus $t$ to ensure correct modular reduction during decryption.
\end{proof}

Together, Theorems~\ref{thm:intra-poly} and~\ref{thm:inter-poly} define the feasible region for $(\beta_\ell, \delta_\ell, m_\ell, U)$ such that homomorphic addition over packed ciphertexts remains both slot-wise carry-free and globally modulus-safe.

\section{FedBit Hardware Architecture}

To address the computational demands of FHE in FL, the \emph{FedBit} framework integrates a client-side FPGA-based accelerator. This dedicated hardware enhances the performance of critical cryptographic operations, allowing for efficient local computation while preserving strict privacy guarantees.

\subsection{Overall Architecture}
We have developed an FPGA-based accelerator connecting to the host via PCIe to support high-throughput, low-latency bidirectional data transfer, as illustrated in Figure~\ref{fig 5:FPGA Accelerator Architecture}. The accelerator consists of three primary components: a controller, external DDR memory, and a Crypto Engine. The controller manages communication and instruction dispatch, while the DDR stores packed model weights and ciphertexts awaiting encryption or decryption. Data is streamed from DDR to the Crypto Engine's on-chip memory via a direct memory access (DMA) engine to maximize throughput.

\begin{figure}[t]
    \centering
    \includegraphics[width=\columnwidth]{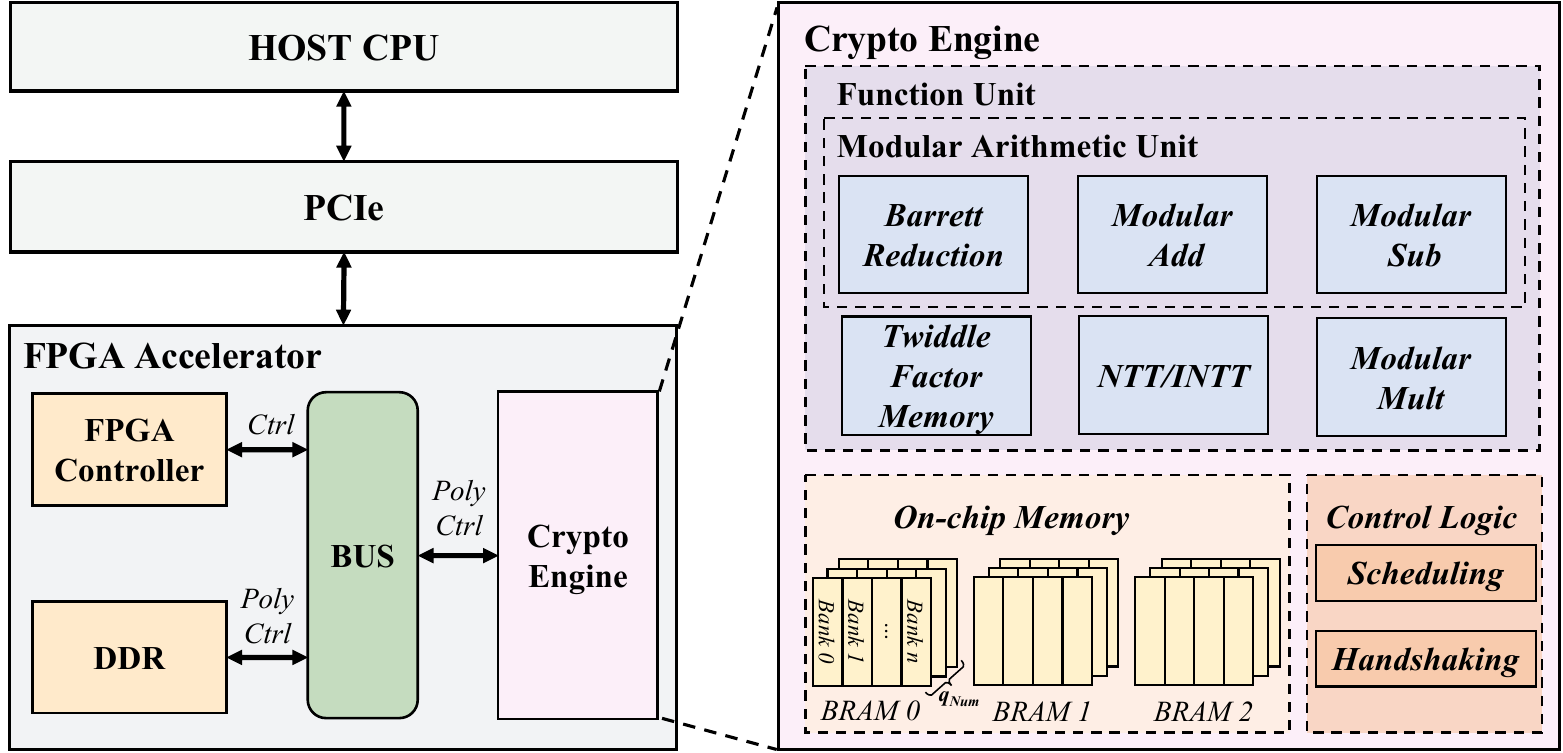}
    \caption{System-Level Accelerator Architecture of \emph{FedBit}}
    \label{fig 5:FPGA Accelerator Architecture}
    \vspace{-15pt}
\end{figure}

The Crypto Engine serves as the computational core, executing all FHE-related polynomial arithmetic. Each ciphertext consists of two degree-\(N\) polynomials. To support large ciphertext moduli, each polynomial coefficient is represented in residue number system (RNS) form and decomposed into \(q_{\text{Num}}\) modular limbs, where each limb corresponds to a distinct modulus in the RNS basis. To efficiently process data,
we adopt a dataflow-oriented micro-architecture inspired by~\cite{meng2024hfntthazardfreedataflowaccelerator}, which enables efficient parallel execution across pipelined stages. Modular multiplication is implemented using a hardware-optimized Barrett reducer, and both NTT and INTT operations are realized through conflict-free, high-throughput pipelines.

Three parallel on-chip BRAM buffers are instantiated to store intermediate polynomials, each organized into \(q_{\text{Num}}\) limbs striped across \(n\) independent banks with \(n\) depth for simultaneous access, where \(n = \sqrt{N}\). All reusable values are cached locally to reduce external memory traffic and latency. 

\subsection{FedBit DataFlow}
\label{subsec:fededge_he_flow}

The accelerator implements encryption and decryption, with the dataflow depicted in Figure~\ref{fig:encryption-dataflow}, which comprises three stages: Preparation, Encryption, and Decryption. 

\begin{enumerate}
\item \textbf{Preparation.} The host generates the secret key \(s\), a random polynomial \(a\), and an error polynomial \(e\), in the coefficient domain. The accelerator then receives these values and precomputes two reusable quantities: \(as + e\) and \(NTT(-a)\).

\item \textbf{Encryption.} The plaintext \(m\) is streamed to the accelerator, scaled by \(\Delta\), and added to the precomputed term \(as + e\) to form the first ciphertext component, while \(NTT(-a)\) forms the second component.

\item \textbf{Decryption.} 
The server aggregates the $a$ in the NTT domain, to minimize the NTT/INTT operations, and the $\Delta m + as + e$ in the coefficient domain from all the clients.
After receiving the aggregated results, the accelerator multiplies \(NTT(-a')\) with the stored secret key \(NTT(s)\), adds the result to the first ciphertext component, and divides by \(\Delta\) to recover the plaintext \(m'\), which is the summation of the messages from the clients.
\end{enumerate}

\begin{figure}[h]
    \centering
    \includegraphics[width=0.95\columnwidth]{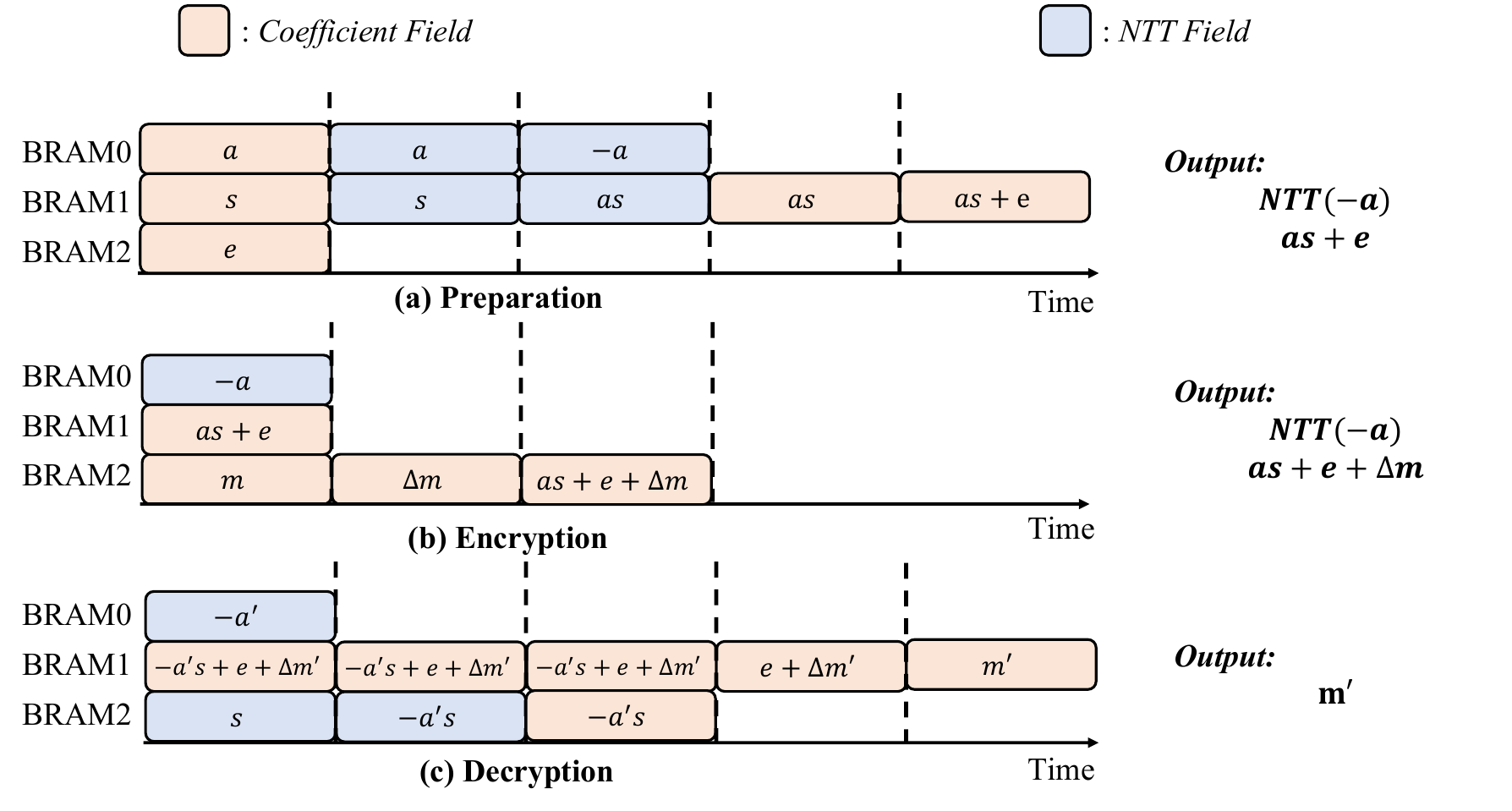}
    \caption{Dataflow of Preparation, Encryption, and Decryption.}
    \label{fig:encryption-dataflow}
    \vspace{-15pt}
\end{figure}

\begin{figure*}[t]
    \centering
    \begin{subfigure}[b]{0.33\textwidth}
        \includegraphics[width=\textwidth]{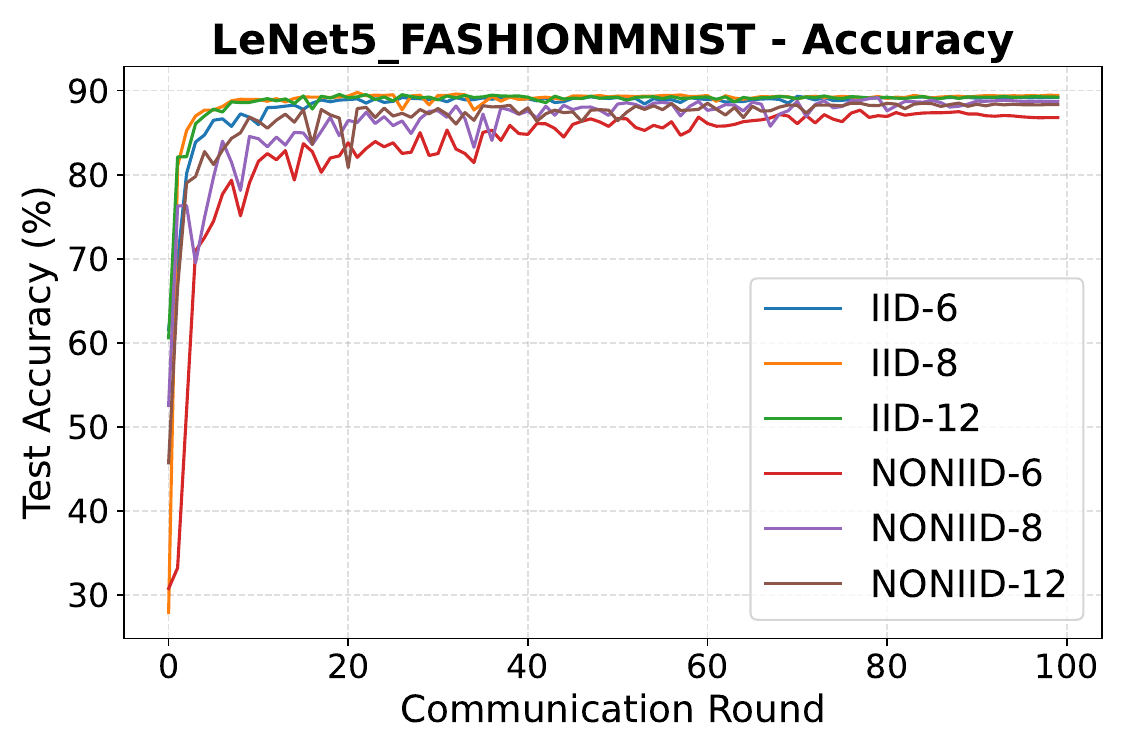}
        \caption{LeNet-5 FashionMNIST}
        \label{fig:lenet5_fashionmnist_accuracy}
    \end{subfigure}
    \hfill
    \begin{subfigure}[b]{0.33\textwidth}
        \includegraphics[width=\textwidth]{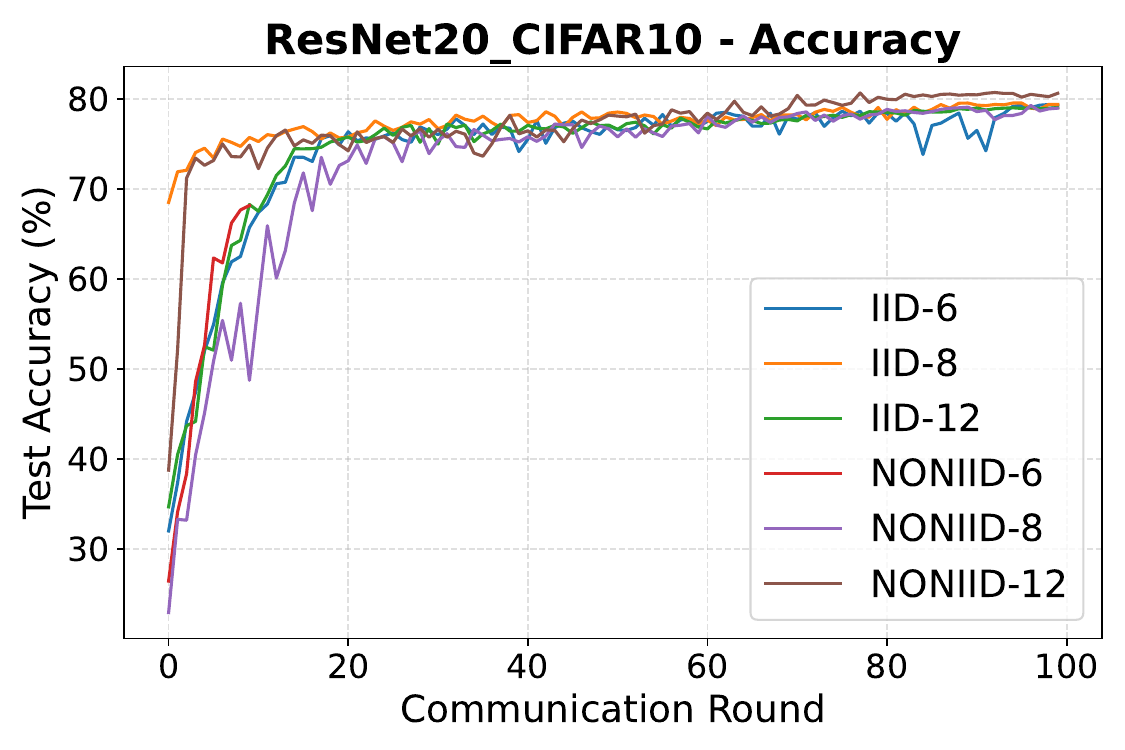}
        \caption{ResNet-20 on CIFAR-10}
        \label{fig:resnet18_cifar10_accuracy}
    \end{subfigure}
    \hfill
    \begin{subfigure}[b]{0.33\textwidth}
        \includegraphics[width=\textwidth]{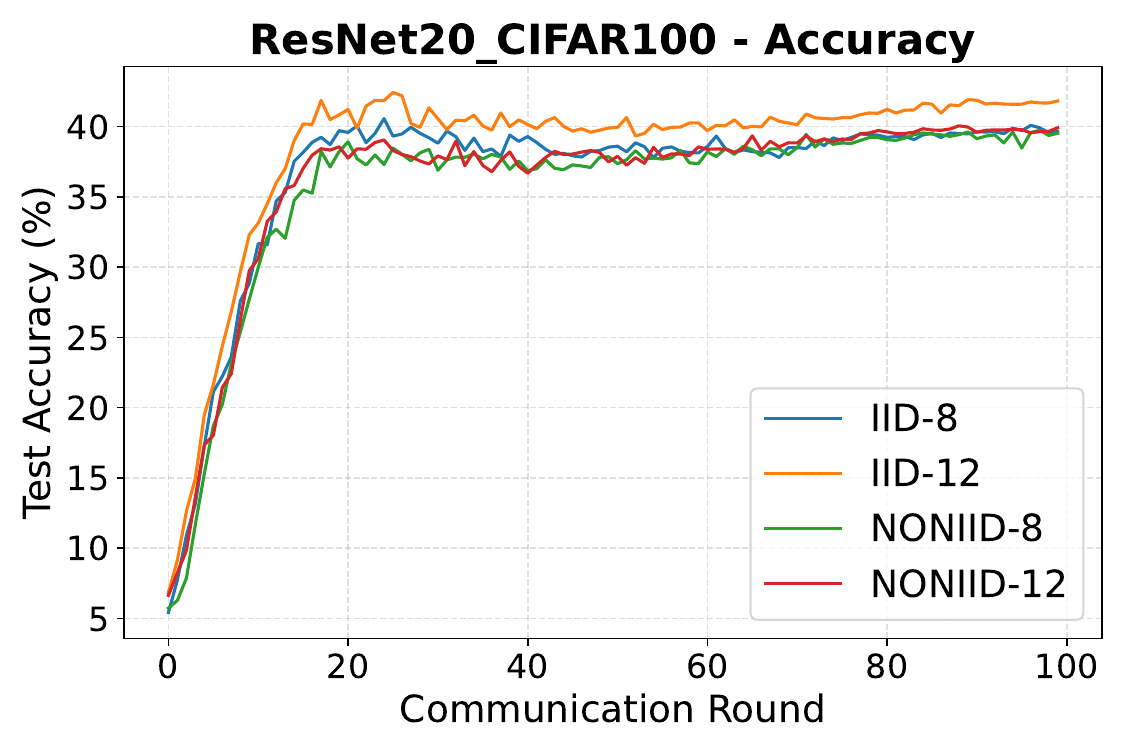}
        \caption{ResNet-20 on CIFAR-100}
        \label{fig:resnet18_cifar100_accuracy}
    \end{subfigure}

    \caption{Accuracy of different models on benchmark datasets over 100 communication rounds.}
    \label{fig:accuracy_comparison}
    \vspace{-15pt}
\end{figure*}
    
\section{\textsc{FedBit} Framework}
\label{sec:fedge-he-framework}

\begin{algorithm}[!hb]
    \caption{\emph{FedBit} Federated Learning Workflow}
    \label{alg:secure_fl}
    {
    \small
    \begin{algorithmic}[1]
    \REQUIRE No. of clients $U$, aggregation clients per round $M$, Rounds $T$
    \ENSURE Global model $\mathbf{W}^{(T)}$
    
    \STATE \textcolor{blue}{\textbf{Phase 1: Initialization}}
    \STATE Broadcast initial model weights $\mathbf{W}^{(0)}$ to all clients
    \STATE All clients agree on a shared secret key $sk$ via a key exchange protocol

    \FOR{round $t = 1$ to $T$}
        \STATE  \textcolor{blue}{\textbf{Phase 2: Local Computation (Clients)}}
        \FOR{each client $i = 1$ to $U$ \textbf{in parallel}}
            \STATE Generate the public key \(pk_i\) according to secret key \(sk\)
            \STATE $\mathbf{W}_i^{(t)} \gets \textsc{LocalTrain}(\mathbf{W}^{(t-1)})$
            \STATE $\mathbf{P}_i^{(t)} \gets \textsc{QuantizeAndPack}(\mathbf{W}_i^{(t)}, \mathbf{W}^{(t-1)})$
            \STATE $\mathbf{E}_i^{(t)} \gets \textsc{EncryptPackedUpdate}(\mathbf{P}_i^{(t)}, pk_i)$ \textcolor{red}{\textbf{[FPGA]}}
            \STATE Send $\mathbf{E}_i^{(t)}$ to server
        \ENDFOR
    
        \STATE \textcolor{blue}{\textbf{Phase 3: Secure Aggregation (Server)}}
        \STATE Wait for $M$ encrypted updates
        \STATE $\mathbf{E}^{(t)} \gets \sum_{i=1}^{M} \mathbf{E}_i^{(t)}$ \COMMENT{Homomorphic sum}
        \STATE Broadcast $\mathbf{E}^{(t)}$ to all clients
    
        \STATE \textcolor{blue}{\textbf{Phase 4: Model Update (Clients)}}
        \FOR{each client $i = 1$ to $U$ \textbf{in parallel}}
            \STATE $\mathbf{P}^{(t)} \gets \textsc{Decrypt}(\mathbf{E}^{(t)}, sk)$ \textcolor{red}{\textbf{[FPGA]}}
            \STATE $\{\mathbf{W}_\ell^{(t)}\}_{\ell=1}^{\mathcal{L}} \gets \textsc{UnpackAndDequantize}(\frac{\mathbf{P}^{(t)}}{M}, \mathbf{W}^{(t-1)})$
            \STATE $\mathbf{W}^{(t)} \gets \textsc{ReconstructModel}(\{\mathbf{W}_\ell^{(t)}\}_{\ell=1}^{\mathcal{L}})$
        \ENDFOR
    \ENDFOR
    \RETURN $\mathbf{W}^{(T)}$
    \end{algorithmic}
    }
\end{algorithm}

The \emph{FedBit} framework enables secure and efficient FL at the edge by integrating FPGA-accelerated homomorphic encryption. It coordinates interactions between $U$ clients and a central server over $T$ communication rounds, ensuring privacy-preserving model updates with minimal computation overhead.

The federated learning protocol proceeds in four phases:

\begin{enumerate}
\item \textbf{Local Training and Encryption:} Each client performs local training and applies layer-wise quantization. For each layer $\ell$, the quantization parameters are computed based on the minimum values $\min(\mathbf{W}_\ell^{(t-1)})$ and maximum values $\max(\mathbf{W}_\ell^{(t-1)})$ of the corresponding weights from the previous round. The quantized weights are then packed and encrypted.

\item \textbf{Secure Aggregation:} The server aggregates encrypted updates from $M \leq U$ clients via homomorphic addition. The aggregated results are broadcast to all the $U$ clients.

\item \textbf{Decryption and Model Update:} Clients decrypt the aggregated ciphertext, unpack and dequantize the weights, and reconstruct the global model.
\end{enumerate}

The complete protocol is formalized in Algorithm~\ref{alg:secure_fl}, where FPGA acceleration is employed during both encryption and decryption to ensure efficient execution.

\section{Experiment}
\label{sec:experiment}

\subsection{Experimental Setup}

\subsubsection*{Hardware and Software Environment}
Experiments were conducted on a CentOS 7 system equipped with an Intel Xeon Gold 6338, 32GB of DDR4 RAM, and an NVIDIA RTX 3080 for local training. The FPGA accelerator is implemented using High-Level Synthesis (HLS) with Vivado HLS 2024.1, specifically targeting the U250 platform. The communication protocol leverages the Flower federated learning framework~\cite{beutel2020flower} for orchestrating the federated training process. 

\subsubsection*{Homomorphic Encryption Configuration}
We implement the BFV scheme with security parameters conforming to contemporary cryptographic standards, with the polynomial degree $N=4096$, plaintext moduli $t=2281701377$, and employ a $128$-bit security level as recommended by the homomorphic encryption standard~\cite{cryptoeprint:2019/939}. 

\subsubsection*{Benchmark Datasets and Models}
We evaluate our framework on three widely-used benchmark datasets: Fashion-MNIST, CIFAR-10, and CIFAR-100, employing LeNet-5 architecture for Fashion-MNIST and ResNet-20 for both CIFAR datasets. All experiments span 100 communication rounds under both independent and identically distributed (IID) and non-IID data distributions to rigorously assess performance under varying degrees of data heterogeneity.

The FL consists of $U=10$ clients, with $M=5$ randomly selected clients participating in each round's secure aggregation process. Each participating client performs one local training epoch per round with a batch size of 64. For optimization, we utilize Adam with an initial learning rate of $10^{-2}$ and implement a cosine annealing scheduler that gradually reduces the learning rate to $10^{-6}$ over the course of training. Following each aggregation step, the server distributes the updated global model to all clients for the subsequent round.

\subsubsection*{Packing Configuration}
For bit-interleaved packing, we utilize a 3-bit for packing, and configure the quantization bit width to 6, 8, and 12 bits for the model weights. 

\subsubsection*{Baselines}
CKKS: A homomorphic encryption scheme supporting approximate arithmetic, implemented with Microsoft SEAL~\cite{sealcrypto}.
BatchCrypt~\cite{zhang2020batchcrypt}: A Paillier-based FL framework evaluated on AWS c5.4xlarge (clients) and r5.4xlarge (aggregator) instances in a geo-distributed setting.
PACK~\cite{10.1145/3698038.3698557}: A communication-efficient FL scheme implemented with TenSEAL, evaluated on a local CPU server (dual Xeon 6226R, 48 GB RAM) and an RTX 3090 GPU server for client-side timing.

\subsection{Accuracy and Communication Overhead}

\begin{table}[b]
    \centering
    \caption{Comparison of traffic and accuracy across different schemes. Traffic is computed per round as the average of 100 rounds for a single client, counting both upload and download. 
    \textsuperscript{*}BatchCrypt results are based on IID data.}
    \setlength{\tabcolsep}{3.5pt} 
    \begin{tabular}{@{}cccccc@{}}
    \toprule
    Dataset & Metric  & CKKS & BatchCrypt\textsuperscript{*} & PACK & \emph{FedBit} \\
    \midrule
    \multirow{2}{*}{Fashion MNIST}
        & Traffic (MB) & 8.68  & 32.27  & 2.78  & 2.7 \\  
        \cmidrule(lr){2-6} 
        & Accuracy     & 83.9\% & 86.2\% & 84.7\% & 88.81\% \\
    \midrule
    \multirow{2}{*}{CIFAR-10}
        & Traffic (MB) & 1526.53 & 3043.84 & 487.17 & 356.7 \\  
        \cmidrule(lr){2-6} 
        & Accuracy     & 65.1\% & 48.6\% & 67.8\% & 80.72\% \\
    \midrule    
    \multirow{2}{*}{CIFAR-100}
        & Traffic (MB) & 1532.67 & 3043.84 & 489.31 & 357.72 \\  
        \cmidrule(lr){2-6} 
        & Accuracy     & 25.7\% & 13\% & 25.8\% & 40.06\% \\
    \bottomrule
    \end{tabular}
 
    \label{tab:communication_accuracy_comparison}
\end{table}

Figure~\ref{fig:accuracy_comparison} illustrates the accuracy across benchmark datasets over 100 rounds, comparing IID and Non-IID data distributions where the numbers indicate the quantization bit width. All models demonstrate a consistent upward trend in accuracy. Generally, non-IID accuracies are lower than IID ones, and higher quantization bit widths lead to better accuracy. For FashionMNIST, the highest accuracy achieved was 89.17\% with 12-bit quantization, while on CIFAR-10, it achieved 80.72\% with 12-bit quantization, and on the CIFAR-100 dataset, it achieved 40.06\% with 12-bit quantization. The results indicate that the \emph{FedBit} framework effectively maintains model performance across different data distributions and quantization levels.

Table~\ref{tab:communication_accuracy_comparison} provides a comprehensive comparison of communication overhead and model accuracy across different schemes. Under Non-IID data distribution with 12-bit quantization, \emph{FedBit} demonstrates superior communication efficiency, reducing average per-client traffic by 54.3\%, 63.9\%, and 63.9\% for Fashion-MNIST, CIFAR-10, and CIFAR-100 datasets compared with three baselines. 

More importantly, \emph{FedBit} consistently achieves substantial accuracy improvements across all benchmarks. On Fashion-MNIST, \emph{FedBit} outperforms CKKS, BatchCrypt, and PACK by 4.91\%, 2.61\%, and 4.11\%, respectively. The performance gap becomes more pronounced on complex datasets: on CIFAR-10, \emph{FedBit} surpasses CKKS by 15.62\%, BatchCrypt by 32.12\%, and PACK by 12.92\%; on CIFAR-100, the gains are 14.36\%, 27.06\%, and 14.26\%, respectively. This is primarily attributed to the 12-bit quantization of model weights, which constrains parameter values to a fixed-point range and thus prevents overfitting. Furthermore, the use of the BFV scheme, which supports exact integer arithmetic, ensures greater numerical stability and avoids the noise amplification inherent in CKKS's approximate arithmetic, which can otherwise lead to degraded precision over successive computations.

\subsection{Hardware Resource and Time Analysis}

\begin{figure}[!t]
\begin{tikzpicture}
    \begin{axis}[
        width=\columnwidth,
        height=5cm,
        ybar=2pt,
        bar width=8pt,
        ylabel={Time (ms)},
        ylabel style={font=\small},
        ymode=log,
        log basis y={10},
        ymin=1,
        ymax=10000000,
        ytick={1,10,100,1000,10000,100000,1000000,10000000},
        yticklabel style={font=\small},
        scaled y ticks=false,
        xtick={1,2,3,4},
        xticklabels={
            {Enc.},
            {Dec.},
            {Enc.},
            {Dec.}
        },
        xticklabel style={font=\small},
        xlabel style={font=\small},
        legend style={
            at={(0.5,1.02)},
            anchor=south,
            legend columns=4,
            font=\small,
            draw=black,
            fill=white,
            column sep=2mm,
            legend cell align=left
        },
        grid=major,
        grid style={line width=0.2pt, draw=gray!30},
        major grid style={line width=0.2pt, draw=gray!30},
        enlarge x limits=0.15,
        axis line style={line width=0.4pt},
        tick style={line width=0.4pt},
        after end axis/.code={
            \node[font=\small, anchor=north] at (axis description cs:0.25,-0.15) {Fashion-MNIST};
            \node[font=\small, anchor=north] at (axis description cs:0.75,-0.15) {CIFAR-10};
        },
        clip=false,  
    ]
    
    \addplot[
        fill=myblue,
        draw=black,
        line width=0.4pt,
    ] coordinates {
        (1, 2862)
        (2, 15)
        (3, 492930)
        (4, 2770)
    };

    \node[font=\tiny, anchor=south] at (axis cs:0.7,2862) {2.9k};
    \node[font=\tiny, anchor=south] at (axis cs:1.7,15) {$15$};
    \node[font=\tiny, anchor=south] at (axis cs:2.7,492930) {492k};
    \node[font=\tiny, anchor=south] at (axis cs:3.7,2770) {2.7k};

    \addplot[
        fill=mygreen,
        draw=black,
        line width=0.4pt,
    ] coordinates {
        (1, 2336)
        (2, 7)
        (3, 414060)
        (4, 1380)
    };

    \node[font=\tiny, anchor=south] at (axis cs:0.9,2336) {2.3k};
    \node[font=\tiny, anchor=south] at (axis cs:1.9,40) {7};
    \node[font=\tiny, anchor=south] at (axis cs:2.9,414060) {414k};
    \node[font=\tiny, anchor=south] at (axis cs:3.9,1200) {1.4k};
    
    \addplot[
        fill=myorange,
        draw=black,
        line width=0.4pt,
    ] coordinates {
        (1, 14177)
        (2, 4873)
        (3, 1642960)
        (4, 788520)
    };

    \node[font=\tiny, anchor=south] at (axis cs:1.1,14177) {14k};
    \node[font=\tiny, anchor=south] at (axis cs:2.1,4873) {4.8k};
    \node[font=\tiny, anchor=south] at (axis cs:3.1,1642960) {1643k};
    \node[font=\tiny, anchor=south] at (axis cs:4.1,788520) {789k};
    
    \addplot[
        fill=myred,
        draw=black,
        line width=0.4pt
    ] coordinates {
        (1, 3.2)
        (2, 3.1)
        (3, 420.4)
        (4, 413.9)
    };
    
    \node[font=\tiny, anchor=south] at (axis cs:1.3,3.2) {3.2};
    \node[font=\tiny, anchor=south] at (axis cs:2.3,3.1) {3.1};
    \node[font=\tiny, anchor=south] at (axis cs:3.3,421) {0.4k};
    \node[font=\tiny, anchor=south] at (axis cs:4.3,414) {0.4k};
    
    \legend{CKKS, PACK, BatchCrypt, \emph{FedBit}}
    \draw[dashed, gray!50, line width=1pt] (axis cs:1.5,1) -- (axis cs:1.5,10000000);
    \draw[dashed, gray!50, line width=1pt] (axis cs:2.5,1) -- (axis cs:2.5,10000000);
    \draw[dashed, gray!50, line width=1pt] (axis cs:3.5,1) -- (axis cs:3.5,10000000);
    
    \end{axis}
\end{tikzpicture}
\caption{Per-round encryption and decryption time for different HE-FL schemes on Fashion-MNIST and CIFAR-10.}
\label{fig:Time}
\vspace{-15pt}
\end{figure}
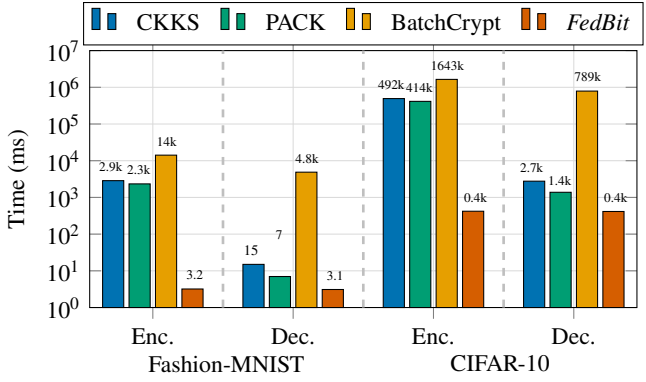

To evaluate the practical deployment feasibility of \emph{FedBit}, we implemented our hardware accelerator on the Xilinx Alveo U250 FPGA using HLS. Our optimized design achieves efficient resource utilization with only 2\% Block RAM, 7\% DSP slices, 36\% Flip-Flops, and 71\% LUTs while operating at 250 MHz.

\subsubsection{Encryption and Decryption Performance}

Figure~\ref{fig:Time} reveals the substantial speedup achieved by \emph{FedBit} compared to state-of-the-art HE-FL implementations. For the FashionMNIST dataset, encryption and decryption times are reduced to $3.2$ ms and $3.1$ ms, respectively, representing speedups of $894\times$ over CKKS and $4,430\times$ over BatchCrypt for encryption. Even for the more complex CIFAR-10 dataset, \emph{FedBit} maintains sub-second performance with 0.42 seconds for both encryption and decryption, achieving over three orders of magnitude improvement compared to BatchCrypt's 1,643 seconds encryption time.

\subsubsection{System-Level Time Breakdown}
Figure~\ref{fig:time_breakdown_pct} demonstrates the crucial role of hardware acceleration in making homomorphic encryption feasible for federated learning.

Without hardware acceleration, encryption operations become the dominant computational burden, particularly for complex models. For CIFAR-10 and CIFAR-100, encryption consumes over half of the total execution time when running on a CPU. This long encryption time makes the system impractical for real-world deployment. Even for the simpler Fashion-MNIST dataset, encryption still requires 7.46\% of the total time, which becomes significant when scaled across thousands of training rounds.

\begin{figure}[!t]
    \centering
    \begin{tikzpicture}
      
        \begin{axis}[
            ybar stacked,
            width=\columnwidth,
            height=6cm,
            symbolic x coords={Fashion-MNIST CPU, Fashion-MNIST Accel., CIFAR-100 CPU, CIFAR-100 Accel., CIFAR-10 CPU, CIFAR-10 Accel.},
            xtick=data,   
            xticklabel style={
                font=\footnotesize,
                rotate=45,
                anchor=north east
            },
            ymin=0,
            ymax=110,
            ylabel={Percentage (\%)},
            ylabel style={font=\small},
            yticklabel style={font=\small},
            ytick={0,20,40,60,80,100},
            legend style={
                at={(0.5,1.05)},
                anchor=south,
                legend columns=4,
                font=\footnotesize,
                draw=none,
                column sep=1mm,
                row sep=1mm,
            },
            bar width=0.6cm,
            enlarge x limits=0.1,
            extra x tick style={
                grid=major,
                grid style={thick, black!50},
                tick style={draw=none}
            },
        ]

        \addplot[fill=myblue] coordinates {
            (Fashion-MNIST CPU, 91.27)
            (Fashion-MNIST Accel., 98.58)
            (CIFAR-100 CPU, 39.83)
            (CIFAR-100 Accel., 84.77)
            (CIFAR-10 CPU, 38.4)
            (CIFAR-10 Accel., 80.3)
        };
        
        \addplot[fill=mygreen] coordinates {
            (Fashion-MNIST CPU, 0.29)
            (Fashion-MNIST Accel., 0.31)
            (CIFAR-100 CPU, 2.26)
            (CIFAR-100 Accel., 4.84)
            (CIFAR-10 CPU, 2.3)
            (CIFAR-10 Accel., 5.1)
        };
        
        \addplot[fill=myorange] coordinates {
            (Fashion-MNIST CPU, 0.05)
            (Fashion-MNIST Accel., 0.05)
            (CIFAR-100 CPU, 0.34)
            (CIFAR-100 Accel., 0.75)
            (CIFAR-10 CPU, 0.4)
            (CIFAR-10 Accel., 0.8)
        };
        
        \addplot[fill=myred] coordinates {
            (Fashion-MNIST CPU, 7.46)
            (Fashion-MNIST Accel., 0.34)
            (CIFAR-100 CPU, 50.7)
            (CIFAR-100 Accel., 4.26)
            (CIFAR-10 CPU, 51.9)
            (CIFAR-10 Accel., 4.5)
        };
        
        \addplot[fill=mypurple] coordinates {
            (Fashion-MNIST CPU, 0.3)
            (Fashion-MNIST Accel., 0.34)
            (CIFAR-100 CPU, 1.77)
            (CIFAR-100 Accel., 0.37)
            (CIFAR-10 CPU, 1.7)
            (CIFAR-10 Accel., 4)
        };
        
        \addplot[fill=myyellow] coordinates {
            (Fashion-MNIST CPU, 0.58)
            (Fashion-MNIST Accel., 0.33)
            (CIFAR-100 CPU, 4.8)
            (CIFAR-100 Accel., 4.2)
            (CIFAR-10 CPU, 4.9)
            (CIFAR-10 Accel., 4.4)
        };
        
        \addplot[fill=mycyan] coordinates {
            (Fashion-MNIST CPU, 0.02)
            (Fashion-MNIST Accel., 0.02)
            (CIFAR-100 CPU, 0.13)
            (CIFAR-100 Accel., 0.26)
            (CIFAR-10 CPU, 0.1)
            (CIFAR-10 Accel., 0.3)
        };
        
        \addplot[fill=mygray] coordinates {
            (Fashion-MNIST CPU, 0.04)
            (Fashion-MNIST Accel., 0.03)
            (CIFAR-100 CPU, 0.23)
            (CIFAR-100 Accel., 0.55)
            (CIFAR-10 CPU, 0.2)
            (CIFAR-10 Accel., 0.6)
        };

        \draw[dashed, gray!60, line width=1pt] (rel axis cs:0.325, 0) -- (rel axis cs:0.325, 1);
        \draw[dashed, gray!60, line width=1pt] (rel axis cs:0.65, 0) -- (rel axis cs:0.65, 1);
        
        \node[font=\scriptsize, anchor=south] at (axis cs:Fashion-MNIST CPU,100) {1.08s};
        \node[font=\scriptsize, anchor=south] at (axis cs:Fashion-MNIST Accel.,100) {0.94s};
        \node[font=\scriptsize, anchor=south] at (axis cs:CIFAR-100 CPU,100) {21.29s};
        \node[font=\scriptsize, anchor=south] at (axis cs:CIFAR-100 Accel.,100) {9.88s};
        \node[font=\scriptsize, anchor=south] at (axis cs:CIFAR-10 CPU,100) {21.31s};
        \node[font=\scriptsize, anchor=south] at (axis cs:CIFAR-10 Accel.,100) {9.26s};

        \legend{Training, Quantization, Packing, Encryption, Aggregation, Decryption, Unpacking, Dequantization}
        \end{axis}
    \end{tikzpicture}
    \caption{FL Time breakdown for one round of one client. 
    }
    \label{fig:time_breakdown_pct}
    \vspace{-15pt}
\end{figure}
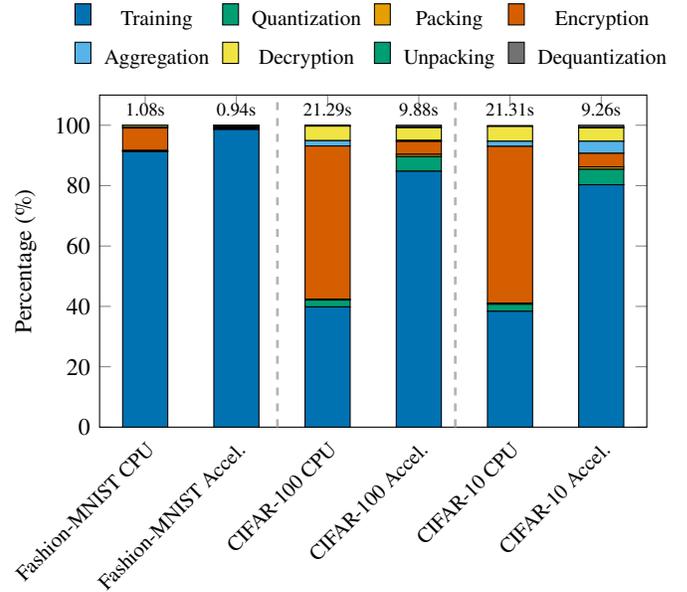

In contrast, \emph{FedBit} with hardware acceleration dramatically alters this scenario. For CIFAR-10, the proportion of time dedicated to encryption plummets from 51.9\% to just 4.5\%. The transformation is equally impressive for CIFAR-100, where encryption overhead drops from 50.7\% to 4.26\%. For Fashion-MNIST, the already modest encryption overhead of 7.46\% is further reduced to a negligible 0.34\%. For CIFAR-10, the total round time decreases from 21.31 seconds to 9.26 seconds, a $2.3\times$ speedup. CIFAR-100 shows similar improvements, dropping from 21.29 to 9.88 seconds. While Fashion-MNIST exhibits a smaller speedup, from 1.08s to 0.94s, due to its inherently lower computational complexity, the relative improvement in encryption efficiency remains substantial.

\section{Conclusion}
\label{sec:conclusion}


The deployment of FL requires rigorous privacy guarantees, yet the computational and bandwidth overhead of homomorphic encryption hinders practical deployment. \emph{FedBit} tackles this barrier by coupling the BFV scheme with bit-interleaved plaintext packing to shrink ciphertext traffic and by offloading encryption and decryption to a client-side accelerator. This hardware-software co-design cuts combined encryption and decryption latency by roughly two orders of magnitude, decreases the traffic overhead by 67\%, and boosts accuracy by about 10\% over baselines, enabling scalable HE-protected FL for large models.

\bibliographystyle{IEEEtran}  
\bibliography{ref}  

\begin{thebibliography}{10}
\providecommand{\url}[1]{#1}
\csname url@samestyle\endcsname
\providecommand{\newblock}{\relax}
\providecommand{\bibinfo}[2]{#2}
\providecommand{\BIBentrySTDinterwordspacing}{\spaceskip=0pt\relax}
\providecommand{\BIBentryALTinterwordstretchfactor}{4}
\providecommand{\BIBentryALTinterwordspacing}{\spaceskip=\fontdimen2\font plus
\BIBentryALTinterwordstretchfactor\fontdimen3\font minus \fontdimen4\font\relax}
\providecommand{\BIBforeignlanguage}[2]{{%
\expandafter\ifx\csname l@#1\endcsname\relax
\typeout{** WARNING: IEEEtran.bst: No hyphenation pattern has been}%
\typeout{** loaded for the language `#1'. Using the pattern for}%
\typeout{** the default language instead.}%
\else
\language=\csname l@#1\endcsname
\fi
#2}}
\providecommand{\BIBdecl}{\relax}
\BIBdecl

\bibitem{truong2021privacy}
N.~Truong, K.~Sun, S.~Wang, F.~Guitton, and Y.~Guo, ``Privacy preservation in federated learning: An insightful survey from the gdpr perspective,'' \emph{Computers \& Security}, vol. 110, p. 102402, 2021.

\bibitem{gdpr2016}
\BIBentryALTinterwordspacing
``Regulation (eu) 2016/679 of the european parliament and of the council of 27 april 2016 on the protection of natural persons with regard to the processing of personal data and on the free movement of such data (general data protection regulation),'' 2016, official Journal of the European Union, L 119/1. [Online]. Available: \url{https://eur-lex.europa.eu/legal-content/EN/TXT/?uri=CELEX:32016R0679}
\BIBentrySTDinterwordspacing

\bibitem{hipaa1996}
\BIBentryALTinterwordspacing
{United States Congress}, ``Health insurance portability and accountability act of 1996,'' 1996, pub.L. 104–191, 110 Stat. 1936. [Online]. Available: \url{https://www.congress.gov/104/plaws/publ191/PLAW-104publ191.pdf}
\BIBentrySTDinterwordspacing

\bibitem{melis2019exploiting}
L.~Melis, C.~Song, E.~De~Cristofaro, and V.~Shmatikov, ``Exploiting unintended feature leakage in collaborative learning,'' in \emph{2019 IEEE symposium on security and privacy (SP)}.\hskip 1em plus 0.5em minus 0.4em\relax IEEE, 2019, pp. 691--706.

\bibitem{nasr2018comprehensive}
M.~Nasr, R.~Shokri, and A.~Houmansadr, ``Comprehensive privacy analysis of deep learning,'' in \emph{Proceedings of the 2019 IEEE Symposium on Security and Privacy (SP)}, vol. 2018, 2018, pp. 1--15.

\bibitem{zhang2023advancing}
B.~Zhang, G.~Lu, P.~Qiu, X.~Gui, and Y.~Shi, ``Advancing federated learning through verifiable computations and homomorphic encryption,'' \emph{Entropy}, vol.~25, no.~11, p. 1550, 2023.

\bibitem{pan2024fedshe}
Y.~Pan, Z.~Chao, W.~He, Y.~Jing, L.~Hongjia, and W.~Liming, ``Fedshe: privacy preserving and efficient federated learning with adaptive segmented ckks homomorphic encryption,'' \emph{Cybersecurity}, vol.~7, no.~1, p.~40, 2024.

\bibitem{cai2023secfed}
Y.~Cai, W.~Ding, Y.~Xiao, Z.~Yan, X.~Liu, and Z.~Wan, ``Secfed: A secure and efficient federated learning based on multi-key homomorphic encryption,'' \emph{IEEE Transactions on Dependable and Secure Computing}, vol.~21, no.~4, pp. 3817--3833, 2023.

\bibitem{gentry2009fully}
C.~Gentry, ``Fully homomorphic encryption using ideal lattices,'' in \emph{Proceedings of the forty-first annual ACM symposium on Theory of computing}, 2009, pp. 169--178.

\bibitem{sealcrypto}
``{M}icrosoft {SEAL} (release 4.1),'' \url{https://github.com/Microsoft/SEAL}, Jan. 2023, microsoft Research, Redmond, WA.

\bibitem{10.1145/3698038.3698557}
\BIBentryALTinterwordspacing
Z.~Zuo, N.~Su, B.~Li, and T.~Zhang, ``Pack: Towards communication-efficient homomorphic encryption in federated learning,'' in \emph{Proceedings of the 2024 ACM Symposium on Cloud Computing}, ser. SoCC '24.\hskip 1em plus 0.5em minus 0.4em\relax New York, NY, USA: Association for Computing Machinery, 2024, p. 470–486. [Online]. Available: \url{https://doi.org/10.1145/3698038.3698557}
\BIBentrySTDinterwordspacing

\bibitem{zhang2020batchcrypt}
C.~Zhang, S.~Li, J.~Xia, W.~Wang, F.~Yan, and Y.~Liu, ``Batchcrypt: Efficient homomorphic encryption for cross-silo federated learning,'' in \emph{2020 USENIX annual technical conference (USENIX ATC 20)}, 2020, pp. 493--506.

\bibitem{cryptoeprint:2024/1976}
\BIBentryALTinterwordspacing
F.~Chen, J.~Dong, X.~Hu, Z.~Dong, and W.~Dai, ``{HI}-{CKKS}: Is high-throughput neglected? reimagining {CKKS} efficiency with parallelism,'' Cryptology {ePrint} Archive, Paper 2024/1976, 2024. [Online]. Available: \url{https://eprint.iacr.org/2024/1976}
\BIBentrySTDinterwordspacing

\bibitem{9857557}
B.~Paul, T.~K. Yadav, B.~Singh, S.~Krishnaswamy, and G.~Trivedi, ``A resource efficient software-hardware co-design of lattice-based homomorphic encryption scheme on the fpga,'' \emph{IEEE Transactions on Computers}, vol.~72, no.~5, pp. 1247--1260, 2023.

\bibitem{10374366}
B.~Che, Z.~Wang, Y.~Chen, L.~Guo, Y.~Liu, Y.~Tian, and J.~Zhao, ``Unifl: Accelerating federated learning using heterogeneous hardware under a unified framework,'' \emph{IEEE Access}, vol.~12, pp. 582--598, 2024.

\bibitem{9893091}
Z.~Wang, B.~Che, L.~Guo, Y.~Du, Y.~Chen, J.~Zhao, and W.~He, ``Pipefl: Hardware/software co-design of an fpga accelerator for federated learning,'' \emph{IEEE Access}, vol.~10, pp. 98\,649--98\,661, 2022.

\bibitem{brakerski2014leveled}
Z.~Brakerski, C.~Gentry, and V.~Vaikuntanathan, ``(leveled) fully homomorphic encryption without bootstrapping,'' \emph{ACM Transactions on Computation Theory (TOCT)}, vol.~6, no.~3, pp. 1--36, 2014.

\bibitem{meng2024hfntthazardfreedataflowaccelerator}
\BIBentryALTinterwordspacing
X.~Meng, Z.~Jiang, and Y.~Lyu, ``Hf-ntt: Hazard-free dataflow accelerator for number theoretic transform,'' 2024. [Online]. Available: \url{https://arxiv.org/abs/2410.04805}
\BIBentrySTDinterwordspacing

\bibitem{beutel2020flower}
D.~J. Beutel, T.~Topal, A.~Mathur, X.~Qiu, J.~Fernandez-Marques, Y.~Gao, L.~Sani, H.~L. Kwing, T.~Parcollet, P.~P.~d. Gusmão, and N.~D. Lane, ``Flower: A friendly federated learning research framework,'' \emph{arXiv preprint arXiv:2007.14390}, 2020.

\bibitem{cryptoeprint:2019/939}
\BIBentryALTinterwordspacing
M.~Albrecht, M.~Chase, H.~Chen, J.~Ding, S.~Goldwasser, S.~Gorbunov, S.~Halevi, J.~Hoffstein, K.~Laine, K.~Lauter, S.~Lokam, D.~Micciancio, D.~Moody, T.~Morrison, A.~Sahai, and V.~Vaikuntanathan, ``Homomorphic encryption standard,'' Cryptology {ePrint} Archive, Paper 2019/939, 2019. [Online]. Available: \url{https://eprint.iacr.org/2019/939}
\BIBentrySTDinterwordspacing

\end{thebibliography}

\end{document}